\documentclass[pra,twocolumn,showpacs,superscriptaddress]{revtex4-1}

\usepackage[dvips]{graphicx} \usepackage{amsmath,amsthm,amssymb}

\newcommand{\bra}[1]{\langle#1|} \newcommand{\ket}[1]{|#1\rangle}
\newcommand{\braket}[2]{\langle#1|#2\rangle}
\newcommand{\ketbra}[2]{|#2\rangle\langle#1|}

\newcommand{\Tr}{\operatorname{Tr}}
\newcommand{\hilb}[1]{\mathcal{#1}} \newcommand{\set}[1]{\mathcal{#1}}
\newcommand{\lin}{\operatorname{Lin}}
\newcommand{\conv}{\operatorname{Conv}}

\newcommand{\rank}{\operatorname{rank}} \newcommand{\identity}{I}

\newtheorem{dfn}{Definition} \newtheorem{lmm}{Lemma}
\newtheorem{thm}{Proposition} 
\newtheorem{algo}{Algorithm} \newtheorem{rmk}{Remark}

\begin{document}

\title{Robustness of Device Independent Dimension Witnesses}

\author{Michele \surname{Dall'Arno}}

\affiliation{ICFO-Institut de Ciencies Fotoniques, Mediterranean
  Technology Park, E-08860 Castelldefels (Barcelona), Spain}

\affiliation{Graduate School of Information Science, Nagoya
  University, Chikusa-ku, Nagoya, 464-8601, Japan}

\author{Elsa \surname{Passaro}}

\author{Rodrigo \surname{Gallego}}

\affiliation{ICFO-Institut de Ciencies Fotoniques, Mediterranean
  Technology Park, E-08860 Castelldefels (Barcelona), Spain}

\author{Antonio \surname{Ac\'in}}

\affiliation{ICFO-Institut de Ciencies Fotoniques, Mediterranean
  Technology Park, E-08860 Castelldefels (Barcelona), Spain}

\affiliation{ICREA-Institucio Catalana de Recerca i Estudis
  Avan\c{c}ats, Lluis Companys 23, E-08010 Barcelona, Spain}

\date{\today}

\begin{abstract}
  Device independent dimension witnesses provide a lower bound on the
  dimensionality of classical and quantum systems in a ``black box''
  scenario where only correlations between preparations, measurements
  and outcomes are considered. We address the problem of the
  robustness of dimension witnesses, namely that to witness the
  dimension of a system or to discriminate between its quantum or
  classical nature, even in the presence of loss. We consider the case
  when shared randomness is allowed between preparations and
  measurements and provide a threshold in the detection efficiency
  such that dimension witnessing can still be performed.
\end{abstract}

\maketitle

\section{Introduction}

For several experimental setups, a description that completely
specifies the nature of each device is unsatisfactory. For example, in
a realistic scenario the assumption that the provider of the devices
is fully reliable is often overoptimistic: imperfections unavoidably
affect the implementation, thus turning it away from its ideal
description. A device independent description of an experimental setup
does not make any assumption on the involved devices, which are
regarded as ``black boxes'', while only the knowledge of the
correlations between preparations, measurements and outcomes is
considered. In this scenario, a natural question is whether it is
possible to derive some properties of the non-characterized devices
instead of assuming them, building only upon the knowledge of these
correlations. In general one could be interested in bounding the
dimension of the systems prepared by a non-characterized device; one
could also ask whether a source is intrinsically quantum or can be
described classically.  The framework of device independent dimension
witnesses (DIDWs) provides an effective answer to these questions,
suitable for experimental implementation and for application in
different contexts, such as quantum key distribution (QKD) or quantum
random access codes (QRACs).

DIDWs were first introduced in \cite{BPAGMS08} in the context of
non-local correlations for multi-partite systems. Subsequently, the
problem of DIDWs was related to that of QRACs in \cite{WCD08}, and in
\cite{WP09} it was reformulated from a dynamical viewpoint allowing
one to obtain lower bounds on the dimensionality of the system from
the evolution of expectation values. A general formalism for tackling
the problem of DIDWs in a prepare and measure scenario was recently
developed in \cite{GBHA10}. The derived formalism allows one to
establish lower bounds on the classical and quantum dimension
necessary to reproduce the observed correlations. Shortly after, the
photon experimental implementations followed, making use of
polarization and orbital angular momentum degrees of
freedom~\cite{HGMBAT12} or polarization and spatial
modes~\cite{ABCB11} to generate ensembles of classical and quantum
states, and certifying their dimensionality as well as their quantum
nature.

DIDWs also allow reformulating several applications in a
device-independent framework. For example, dimension witnesses can be
used to share a secret key between two honest parties.
In~\cite{PB11}, the authors present a QKD protocol whose security
against individual attacks in a semi-device independent scenario is
based on DIDWs. The scenario is called semi-device independent because
no assumption is made on the devices used by the honest parties,
except that they prepare and measure systems of a given dimension.
Another application is given by QRACs, that make it possible to encode
a sequence of qubits in a shorter one in such a way that the receiver
of the message can guess any of the original qubits with maximum
probability of success. In \cite{HPZGZ11,PZ10} QRACs were considered
in the semi-device independent scenario, with a view to their
application in randomness expansion protocols.

Clearly any experimental implementation of DIDWs is unavoidably
affected by losses - that can be modeled as a constraint on the
measurements - and can reduce the value of the dimension witness, thus
making it impossible to witness the dimension of a system.  Based on
these considerations, it is relevant to understand whether it is
possible to perform reliable dimension witnessing in realistic
scenarios and, in particular, with non-optimal detection efficiency.
We refer to this problem as the {\em robustness of device independent
  dimension witnesses}. Despite its relevance for experimental
implementations and practical applications, this problem has not been
addressed in previous literature. The aim of this work is to fill this
gap. We consider the case where shared randomness between preparations
and measurements is allowed. Our main result is to provide the
threshold in the detection efficiency that can be tolerated in
dimension witnessing, in the case where one is interested in the
dimension of the system as well as in the case where one's concern is
to discriminate between its quantum or classical nature.

The paper is structured as follows. In Section \ref{sect:didw} we
introduce the sets of quantum and classical correlations and the
concept of dimension witness as a tool to discriminate whether a given
correlation matrix belongs to these sets. Section \ref{sect:sets}
discusses some properties of the sets of classical and quantum
correlations. In Section \ref{sect:robustness} we provide a threshold
in the detection efficiency that is allowed in witnessing the
dimensionality of a system or in discriminating between its classical
or quantum nature, as a function of the dimension of the system. We
summarize our results and discuss some further developments - such as
dimension witnessing in the absence of correlations between
preparations and measurements or entangled assisted dimension
witnessing - in Section \ref{sect:conclusion}.

\section{Device independent dimension witnesses}
\label{sect:didw}

Let us first fix the notation~\cite{NC00}. Given a Hilbert space
$\hilb{H}$, we denote with $\lin{\hilb{H}}$ the space of linear
operators $X : \hilb{H} \to \hilb{H}$. A quantum state in $\hilb{H}$
is represented by a density matrix, namely a positive semi-definite
matrix $\rho \in \lin{\hilb{H}}$ such that $\Tr[\rho] = 1$. Given a
pure state $\ket{\psi} \in \hilb H$, we denote with $\psi :=
\ket{\psi}\bra{\psi}$ the corresponding projector. A set $R =
\{\rho_i\}$ of states is said to be classical when the states commute
pairwise, namely $[\rho_i,\rho_k] = 0$ for any $i,k$. Here, the notion
of classicality has to be understood in an operational sense: in our
scenario, the observed correlations can be reproduced by a classical
variable taking $d$ possible values if, and only if, they can be
reproduced by measurements on pairwise commuting states acting on a
Hilbert space of dimension $d$ (this will become clearer after Lemma
\ref{thm:classicalpovms} below). A general quantum measurement is
represented by a Positive-Operator Valued Measure (POVM), namely a set
of positive semi-definite Hermitian matrices $\Pi^j$ such that
$\sum_j\Pi^j = \identity$. A POVM $\Pi = \{ \Pi^j \}$ is said to be
classical when $[\Pi^j, \Pi^l] = 0$ for any $j,l$. The joint
probability of outcome $j$ given input state $\rho_i$ is given by the
Born rule, namely $p_{j|i} = Tr[\rho_i\Pi^j]$.

The general setup introduced in Ref. \cite{GBHA10} for performing
device independent dimension witnessing is given by a preparing device
(let us say on Alice's side) and a measuring device (on Bob's side) as
in Fig. \ref{fig:setup}.
\begin{figure}[htb]
  \includegraphics[width=0.66\columnwidth]{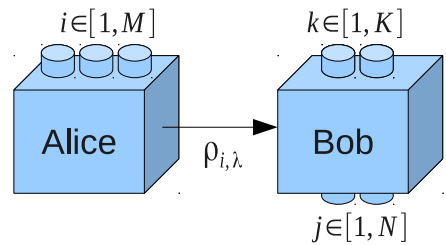}
  \caption{(Color online) Setup for witnessing the dimension of a
    quantum or classical system. In the most general scenario
    considered here, Alice and Bob share a hidden random variable
    $\lambda$. Alice (on the left hand side) owns a preparing device
    which sends the state $\rho_{i,\lambda}$ to Bob whenever Alice
    presses button $i \in [1,M]$. Bob owns a measuring device that
    performs measurement $\Pi_{k,\lambda}$ on the received state
    whenever Bob presses button $k \in [1,K]$, giving the outcome $j
    \in [1,N]$.}
  \label{fig:setup}
\end{figure}
In the most general scenario, the devices may share a priori
correlated information, classical and quantum. However, in many
realistic situations, one can assume that the preparing and measuring
devices are uncorrelated and that all the correlations observed
between the preparation and the measurement are due to the mediating
particle connecting the two devices. An intermediate and also valid
possibility is to assume that the devices only share classical
correlations. In this case, the value of a random variable $\lambda$
distributed according to $q_{\lambda}$ is accessible to preparing and
measuring devices. In this work we focus on this last
possibility. Alice chooses the value of index $i \in [1,M]$ and sends
a fixed state $\rho_{i,\lambda} \in \lin{\hilb{H}}$ to Bob. Bob
chooses the value of index $k \in [1,K]$ and performs a fixed POVM
$\Pi_{k,\lambda}$ on the received state, obtaining outcome $j \in
[1,N]$. After repeating the experiment several times (we consider here
the asymptotic case), they collect the statistics about indexes
$i,j,k$ obtaining the conditional probabilities $p_{j|i,k}$. Note that
we also implicitly assume that we are dealing with independent and
identically distributed events.

We now introduce the set $\set Q$ (the set $\set C$) of correlations
achievable with quantum (classical) preparations.

\begin{dfn}[Set of quantum correlations]
  For any $M,K,N,d \in \mathbb{N}$ we define the {\em set of quantum
    correlations} $\set{Q}(M,K,N,d)$ as the set of correlations
  $p_{j|i,k}$ with $i \in [1,M]$, $k \in [1,K]$ and $j \in [1,N]$ such
  that there exist a Hilbert space $ \hilb{H}$ with $\dim\hilb{H}=d$,
  a quantum set $R = \{ \rho_i \in \lin\hilb{H}\}_1^M$ of states and a
  set $P = \{ \Pi_k \}_1^K$ of POVMs $\Pi_k = \{\Pi_k^j \in
  \lin\hilb{H}\}_1^N$ for which $p_{j|i,k} = \Tr[\rho_i \Pi_k^j]$,
  namely
  \begin{align*}
     \set{Q} := \{p \; | \; & \exists \textrm{ $d$-dimensional Hilbert
       space } \hilb{H},\\ & \exists \textrm{ quantum set } \{ \rho_i
     \in \lin\hilb{H}\}_1^M \textrm{ of states},\\ & \exists \textrm{
       set } \{ \Pi_k \}_1^K \textrm{ of POVMs } \Pi_k = \{\Pi_k^j \in
     \lin\hilb{H}\}_1^N \\ & \textrm{such that } p_{j|i,k} =
     \Tr[\rho_i \Pi_k^j]\}.
  \end{align*}
\end{dfn}

\begin{dfn}[Set of classical correlations]
  For any $M,K,N,d \in \mathbb{N}$ we define the {\em set of classical
    correlations} $\set{C}(M,K,N,d)$ as the set of correlations
  $p_{j|i,k}$ with $i \in [1,M]$, $k \in [1,K]$ and $j \in [1,N]$ such
  that there exist a Hilbert space $ \hilb{H}$ with $\dim\hilb{H}=d$,
  a classical set $R = \{ \rho_i \in \lin\hilb{H}\}_1^M$ of states and
  a set $P = \{ \Pi_k \}_1^K$ of POVMs $\Pi_k = \{\Pi_k^j \in
  \lin\hilb{H}\}_1^N$ for which $p_{j|i,k} = \Tr[\rho_i \Pi_k^j]$,
  namely
  \begin{align*}
     \set{C} := \{p \; | \; & \exists \textrm{ $d$-dimensional Hilbert
       space } \hilb{H},\\ & \exists \textrm{ classical set } \{
     \rho_i \in \lin\hilb{H}\}_1^M \textrm{ of states},\\ & \exists
     \textrm{ set } \{ \Pi_k \}_1^K \textrm{ of POVMs } \Pi_k =
     \{\Pi_k^j \in \lin\hilb{H}\}_1^N \\ & \textrm{such that }
     p_{j|i,k} = \Tr[\rho_i \Pi_k^j]\}.
  \end{align*}
\end{dfn}

We write $\set{Q}$ and $\set{C}$ omitting the parameters $M,K,N,d$
whenever they are clear from the context.

\begin{rmk}\label{rmk:sr}
  We notice that, when shared randomness is allowed between quantum
  (classical) preparations and measurements, the set of achievable
  correlations is given by $\conv\set Q$ ($\conv\set C$), where for
  any set $\set X$ we denote with $\conv\set X$ the convex hull of
  $\set X$.
\end{rmk}

The following Lemma shows that it is not restrictive to consider only
classical POVMs, that is, measurements consisting of commuting
operators, in the definitions of classical correlations.
\begin{lmm} \label{thm:classicalpovms}
  For any correlation $p = \{ p_{j|i,k} \} \in \set C$ there exist a
  classical set $R = \{ \rho_i \}$ of states and a set $Q = \{
  \Lambda_k\}$ of classical POVMs $\Lambda_k = \{ \Lambda_k^j \}$ such
  that $p_{j|i,k} = \Tr[\rho_i \Lambda_k^j]$.
\end{lmm}
\begin{proof}
  By hypothesis there exist a classical set $R = \{ \rho_i \}$ of
  states and a set $P = \{ \Pi_k \}$ of POVMs $\Pi_k = \{ \Pi_k^j \}$
  such that $p_{j|i,k} = \Tr[\rho_i\Pi_k^j]$ for any $i,j,k$.  Take
  $\Lambda_k^j = \sum_i \bra{i} \Pi_k^j \ket{i} \ket{i}\bra{i}$ where
  $\{ \ket{i} \}$ is an orthonormal basis with respect to which the
  $\rho_i$'s are diagonal (it is straightforward to verify that
  $\Lambda_k^j \ge 0$ for any $k,j$ and $\sum_j \Lambda_k^j =
  \identity$ for any $k$). We have $p_{j|i,k} =
  \Tr[\rho_i\Lambda_k^j]$ for any $i,j,k$, which proves the statement.
\end{proof}

Lemma \ref{thm:classicalpovms} thus proves that every set of
probabilities obtained with commuting states can be performed with
classical states and classical POVMs. This clearly implies that
commuting states may be equally regarded as classical variables, and
commuting-element measurements as read-out of classical variables.

We can now introduce DIDWs. Building only on the knowledge of
$p_{j|i,k}$, our task is to provide a lower bound on the dimension $d$
of $\hilb{H}$

\begin{dfn}
  For any set of correlations $\set X$ between $M$ preparations and
  $K$ measurements with $N$ outcomes, a {\em device independent
    dimension witness} $W_{\set X}(p)$ is a function of the
  conditional probability distribution $p = \{ p_{j|i,k} \}$ with $i
  \in [1,M]$, $k \in [1,K]$, and $j \in [1,N]$ such that
  \begin{align}\label{eq:didw}
    W_{\set X}(p) > L \Rightarrow p \not\in \set X,
  \end{align}
  for some $L$ which depends on $W_{\set X}$.
\end{dfn}

Interestingly, in many situations the value of the bound $L$ in the
definition of a dimension witness varies depending on whether one is
interested in classical or quantum ensembles of states. This gives a
second application for dimension witnesses, namely quantum
certification: if the system dimension is assumed, dimension witnesses
allow certifying its quantum nature. It is precisely this quantum
certification what makes dimension witnesses useful for quantum
information protocols \cite{PB11,HPZGZ11}.

Motivated by Remark \ref{rmk:sr}, for any $M,N,K,d \in \mathbb{N}$
when $\set X = \conv\set C(M,N,K,d)$ [when $\set X = \conv\set
  Q(M,N,K,d)$] we say that $W_{\set X}(p)$ is a classical (quantum)
dimension witness for dimension $d$ in the presence of shared
randomness. Given a set $R = \{ \rho_{i,\lambda} \}$ of states and a
set $P = \{ \Pi_{k,\lambda} \}$ of POVMs $\Pi_{k,\lambda} = \{
\Pi_{k,\lambda}^j \}$, we define $W_{\conv\set C}(R,P) := W_{\conv\set
  C} (p)$ with $p = \{ p_{j|i,k} \}$ and $p_{j|i,k} = \sum_{\lambda}
q_{\lambda} \Tr[\rho_{i,\lambda} \Pi_{k,\lambda}^j]$, and analogously
for $W_{\conv\set Q}$.

In this work we will consider only linear DIDWs, namely inequalities
of the form of Eq. \eqref{eq:didw} such that
\begin{align}\label{eq:ldidw}
  W(p) := \vec{c} \cdot \vec{p} = \sum_{i,j,k} c_{i,j,k} p_{j|i,k},
\end{align}
where $\vec{c}$ is a constant vector.

Notice that for any function $W(p)$ and constant $L$, the witness
$W(p) > L$ is only a representative of a class of equivalent witnesses
such that if $W'(p) > L'$ is a member of the class, then $W(p) > L$ if
and only if $W'(p) > L'$ for any conditional distribution $p$. The
following Lemma provides a transformation that preserves this
equivalence.
\begin{lmm}\label{thm:norm}
  Given a function $W(p) = \sum_{i,j,k} c_{i,j,k} p_{j|i,k}$ and a
  constant $L$, take $W'(p) = \sum_{i,j,k} c'_{i,j,k} p_{j|i,k}$ with
  $c'_{i,j,k} = c_{i,j,k} + \alpha_{i,k}$ and $L' = L + \sum_{i,k}
  \alpha_{i,k}$ for any $\alpha_{i,k}$ that does not depend on outcome
  $j$. Then one has $W(p) > L$ if and only if $W'(p) > L'$ for any
  $p$.
\end{lmm}
\begin{proof}
  It follows immediately by direct computation.
\end{proof}

In the following our task will be to find a set $R$ of quantum states
and a set $P$ of POVMs such that a linear witness $W(R,P)$ maximally
violates inequality \eqref{eq:didw}. The following Lemma allows us to
simplify the optimization problem.

\begin{lmm}\label{thm:pure}
  The maximum of any linear dimension witness $W(R,P)$ is achieved by
  an ensemble $R$ of pure states and without shared randomness.
\end{lmm}
\begin{proof}
  The thesis follows immediately from linearity.
\end{proof}

Due to Lemma \ref{thm:pure} the maximization of Eq. \eqref{eq:ldidw}
is equivalent to the maximization of
\begin{align*}
  W(R,P) = \sum_{i,j,k} c_{i,j,k} \bra{\psi_i}\Pi_k^j\ket{\psi_i},
\end{align*}
over the sets $R = \{\psi_i\}$ of pure states and the sets $P = \{
\Pi_k \}$ of POVMs $\Pi_k = \{ \Pi_k^j \}$.

\section{Properties of the sets of quantum and classical
correlations} \label{sect:sets}

Before moving to the main results in this article, we discuss in this
section several properties of the sets of classical and quantum
correlations. In particular, we study whether the sets are convex and
prove some inclusions among them.  These results allow gaining a
better understanding of the geometry of these sets of correlations.

Since classical correlations can always be reproduced by quantum ones,
we immediately have $\set C \subseteq \set Q$ and $\conv\set C
\subseteq \conv\set Q$. Moreover, by definition we have $\set C
\subseteq \conv\set C$ and $\set Q \subseteq \conv\set Q$. Here we
show an example where $\set C$ is non-convex (namely $\set C \subset
\conv \set C$) and $\set C \subset \set Q$. Take $M=3$, $K = 2$,
$N=2$, and $d=2$ in the setup of Figure \ref{fig:setup}.  Consider the
following conditional probability distribution $p_{j|i,k}$ of
obtaining outcome $j$ on Bob's device given input $i$ on Alice's and
$k$ on Bob's
\begin{align}\label{eq:p}
  p_{j|i,1} = \left( \begin{array}{cc} 1 & 0 \\ \frac12 & \frac12 \\ 0
    & 1 \end{array} \right), \qquad p_{j|i,2} =
  \left( \begin{array}{cc} \frac12 & \frac12 \\ 1 & 0 \\ \frac12 &
    \frac12 \end{array} \right),
\end{align}
where rows and columns are labeled by $i$ and $j$, respectively.

First we show that $p \in \conv\set C$. Indeed $p$ can be obtained
when Alice and Bob share classical correlations represented by a
uniformly distributed random variable $\lambda$ taking values $1,2$
making use of the classical set $R = \{\rho_{i,\lambda}\}$ of states
and of the set $P = \{\Pi_{k,\lambda}\}$ of classical POVMs
$\Pi_{k,\lambda} = \{ \Pi_{k,\lambda}^j \}$, with
\begin{align*}
  \rho_{1,1} = \ketbra{0}{0}, \quad \rho_{2,1} = \ketbra{0}{0}, \quad
  \rho_{3,1} = \ketbra{1}{1}, \\ \rho_{1,2} = \ketbra{0}{0}, \quad
  \rho_{2,2} = \ketbra{1}{1}, \quad \rho_{3,2} = \ketbra{1}{1},
\end{align*}
and
\begin{align*}
  \Pi^1_{1,1} = \ketbra{0}{0}, \quad \Pi^1_{2,1} = \ketbra{0}{0},
  \\ \Pi^1_{1,2} = \ketbra{0}{0}, \quad \Pi^1_{2,2} = \ketbra{1}{1},
\end{align*}
which proves that $p=\{p_{j|i,k}=\sum_{\lambda} q_{\lambda}
\Tr[\rho_{i,\lambda} \Pi_{k,\lambda}^j]\} \in \conv\set C$.

Now we show that $p \in \set Q$. Indeed $p$ can be obtained by Alice
and Bob making use of the quantum set $R = \{\rho_i\}$ of states and
of the set $P = \{ \Pi_k \}$ of quantum POVMs $\Pi_k = \{ \Pi_k^j \}$,
with
\begin{align*}
  \rho_1 = \ketbra{0}{0}, \quad \rho_2 = \ketbra{+}{+}, \quad \rho_3 =
  \ketbra{1}{1},
\end{align*}
and
\begin{align*}
  \Pi_1^1 = \ketbra{0}{0}, \quad \Pi_2^1 = \ketbra{+}{+},
\end{align*}
which proves that $p \in \set Q$.

Finally, we verify that if Alice and Bob make use of classical sets of
states and POVMS and do not have access to shared randomness there is
no way to achieve the probability distribution $p$ given by
Eq. \eqref{eq:p}. Indeed, to have perfect discrimination between
$\rho_1$ and $\rho_3$ with POVM $\Pi_1$ [see Eq. \eqref{eq:p}], one
must take $\rho_1$ and $\rho_3$ orthogonal - let us say without loss
of generality $\rho_1 = \ket{0}\bra{0}$ and $\rho_3 = \ket{1}\bra{1}$,
and $\Pi_1^1 = \ket{0}\bra{0}$ and $\Pi_1^2 = \ket{1}\bra{1}$. Due to
the hypothesis of classicality of the sets of states, $\rho_2$ must be
a convex combination of $\rho_1$ and $\rho_3$. Then, in order to have
$p_{j|2,1}$ as in Eq. \eqref{eq:p}, one has to choose $\rho_2 =
(\rho_1 + \rho_3)/2 = \identity/2$. Finally, the only possible choice
for $\Pi_2$ is $\Pi_2^1 = \identity$ and $\Pi_2^1 = 0$, which is
incompatible with the remaining entries of $p_{j|i,2}$ in
Eq. \eqref{eq:p}. This proves that $p \not\in \set C$.

The relations between the sets of quantum and classical correlations
are schematically depicted in Fig.  \ref{fig:sets}.
\begin{figure}[htb]
  \includegraphics[width=0.5\columnwidth]{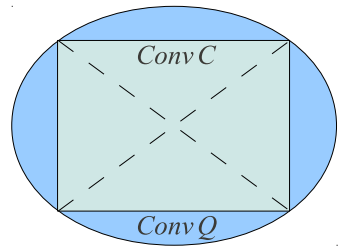}
  \caption{(Color online) Schematic representation of the sets of
    classical and quantum correlations between preparations,
    measurements and outcomes. Dashed line represents the (non-convex)
    set $\set C$ of classical correlations without shared randomness;
    the rectangle represents the set $\conv\set C$ of classical
    correlations with shared randomness; the ellipsoid represents the
    set $\conv\set Q$ of quantum correlations with shared randomness.}
  \label{fig:sets}
\end{figure}

\section{Robustness of dimension witnesses}
\label{sect:robustness}

In practical applications, losses (due to imperfections in the
experimental implementations or artificially introduced by a malicious
provider) can noticeably affect the effectiveness of dimension
witnessing. The main result of this Section is to provide a threshold
value for the detection efficiency that allows one to witness the
dimension of the systems prepared by a source or to discriminate
between its quantum or classical nature.

The task is to determine whether a given conditional probability
distribution belongs to a particular convex set, namely $\conv\set C$
or $\conv\set Q$ (see Remark \ref{rmk:sr}). The situation is
illustrated in Figure \ref{fig:task}.
\begin{figure}[htb]
  \includegraphics[width=0.5\columnwidth]{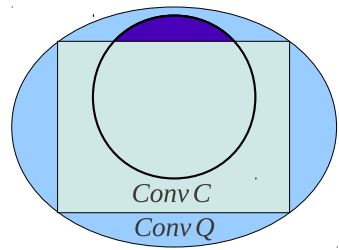}
  \caption{(Color online) The Figure illustrates the problem of the
    robustness of device independent dimension witness. The convex
    hulls $\conv\set Q$ and $\conv\set C$ of the sets of quantum and
    classical correlations are represented as in Figure
    \ref{fig:sets}. In the presence of loss, only a subset of the
    possible correlations is attainable. The subset, surrounded by
    bold line in the figure, is parametrized by detection efficiency
    $\eta$. The task is to find the threshold value in $\eta$ such
    that dimension witnessing is still possible. For example, when the
    task is to discriminate between the quantum or classical nature of
    a source, one is interested in achieving correlations in the dark
    area of the Figure, and our goal is to determine the values of
    $\eta$ such that this area is not null.}
  \label{fig:task}
\end{figure}
The experimental implementation is constrained to be lossy, namely it
can be modeled considering an ideal preparing device followed by a
measurement device with non-ideal detection efficiency. This means
that any POVM $\Pi_{k,\lambda}$ on Bob's side is replaced by a POVM
$\Pi_{k,\lambda}^{(\eta)}$ with detection efficiency $\eta$, namely
\begin{align}\label{eq:lossypovm}
  \Pi_{k,\lambda}^{(\eta)} := \{ \eta\Pi_{k,\lambda},
  (1-\eta)\identity\}.
\end{align}
We notice that each lossy POVM has one outcome more than the ideal
one, corresponding to the no-click event. In a general model, the
detection efficiency $\eta$ may be different for any POVM
$\Pi_{k,\lambda}$. Nevertheless, in the following we assume that they
have the same detection efficiency, which is a reasonable assumption
if the detectors have the same physical implementation
\cite{note1}. Analogously given a set $P = \{ \Pi_{k,\lambda} \}$ of
POVMs we will denote with $P^{(\eta)} = \{ \Pi_{k,\lambda}^{(\eta)}
\}$ the corresponding set of lossy POVMs. Upon defining $p^{(\eta)} :=
\{ p_{j|i,k}^{(\eta)}\}$ with $p_{j|i,k}^{(\eta)} = \sum_{\lambda}
q_\lambda \Tr[\rho_{i,\lambda} \Pi_{k,\lambda}^{j,(\eta)}]$, one
clearly has
\begin{align}\label{eq:lossydist}
  p^{(\eta)} = \eta p^{(1)} + (1-\eta)p^{(0)}.
\end{align}

To attain our task we maximize a given dimension witness over the set
of lossy POVMs as given by Eq. \eqref{eq:lossypovm}. Due to the model
of loss introduced in Eq. \eqref{eq:lossypovm} and to the freedom in
the normalization of dimension witnesses given by Lemma
\ref{thm:norm}, in the following without loss of generality for any
dimension witness $W$ as given in Eq. \eqref{eq:ldidw} it is
convenient to take
\begin{align}\label{eq:norm}
  c_{i,N,k} = 0, \quad \forall i,k.
\end{align}

Then we have the following Lemma.
\begin{lmm}\label{thm:lossy}
  Given a set $R = \{\rho_{i,\lambda}\}_{i=1}^M$ of states and a set
  $P = \{ \Pi_{k,\lambda} \}_{k=1}^K$ of POVMs $\Pi_{k,\lambda} =
  \{\Pi_{k,\lambda}^j\}_{j=1}^{N-1}$, for any linear dimension witness
  $W(p) = \sum_{i,j,k} c_{i,j,k} p_{j|i,k}$ with $i \in [1,M]$, $j \in
       [1,N]$, and $k \in [1,K]$ normalized as in Eq. \eqref{eq:norm}
       one has
  \begin{align*}
    W(R, P^{(\eta)}) = \eta W(R, P^{(1)}).
  \end{align*}
\end{lmm}
\begin{proof}
  One has
  \begin{align*}
    W(R,P^{(\eta)}) & = \sum_{i,j,k} c_{i,j,k} \left[\eta
      p_{j|i,k}^{(1)} + (1-\eta)p_{j|i,k}^{(0)}\right]\\ & = \eta \,
    W(R,P^{(1)}),
  \end{align*}
  where the first equality follows from Eq. \eqref{eq:lossydist} and
  the second from the fact that $W(p^{(0)}) = 0$ due to the
  normalization given in Eq. \eqref{eq:norm}.
\end{proof}

In particular from Lemma \ref{thm:lossy} it follows that for any
linear dimension witness $W$ one has
\begin{align*}
  \max_{R, P} W(R,P^{(\eta)}) = \eta \max_{R,P} W(R,P^{(1)}),
  \\ \arg\max_{R, P} W(R,P^{(\eta)}) = \arg\max_{R,P} W(R,P^{(1)}).
\end{align*}

Due to Lemma \ref{thm:lossy}, it is possible to recast the
optimization of dimension witnesses in the presence of loss to the
optimization in the ideal case. Then due to Lemma \ref{thm:pure} it is
not restrictive to carry out the optimization with pure states and no
shared randomness. Consider the case where $M = d+1$, $K = d$, and $N
= 3$. Using the technique discussed in Appendix \ref{sect:optW} one
can verify that the witness given by Eq. \eqref{eq:ldidw} with the
following coefficients
\begin{align}\label{eq:Id}
  c_{i,j,k} = \left\{ \begin{array}{ll} -1 & \textrm{ if } i+k \le M,
    \; j=1\\ +1 & \textrm{ if } i+k = M+1, \; j=1\\ 0 & \textrm{
      otherwise}
  \end{array}  \right.,
\end{align}
is the most robust to non-ideal detection efficiency. This fact should
not be surprising, as we notice that this witness relies on only $2$
out of $3$ outcomes. According to \cite{GBHA10}, we denote it
$I_{d+1}$. In \cite{GBHA10} (see also \cite{Mas02}) it was conjectured
that for any dimension $d$ the dimension witness $I_{d+1}$ is tight in
the absence of loss.

Now we provide upper and lower bounds for the maximal value $I_{d+1}^*
:= \max_{R,P} I_{d+1}$ where the maximization is over any set $R = \{
\rho_i \in \lin\hilb H \}$ of states and any set $P = \{ \Pi_k \}$ of
POVMs $\Pi_k = \{ \Pi_k^j \in \lin\hilb H\}$ with $\dim\hilb H = d$.
\begin{lmm}\label{thm:recbound}
  For any dimension $d$ we have $I_{d+1}^* \ge I_d^* + 1$.
\end{lmm}
\begin{proof}
  The statement follows from the recursive expression $I_{d+1} = I_d +
  C$, where
  \begin{align*}
    C := -\sum_{i=1}^d \bra{\psi_i} \Pi_1^1 \ket{\psi_i} +
    \bra{\psi_{d+1}} \Pi_1^1 \ket{\psi_{d+1}},
  \end{align*}
  and noticing that $I_d$ and $C$ can be optimized independently.
\end{proof}

A tight upper bound for $I_3$ was provided in \cite{GBHA10}. In the
following Lemma we provide a constructive proof suitable for
generalization to higher dimensions.
\begin{lmm}\label{thm:2bound}
  For dimension $d=2$ we have $I_3^* = \sqrt{2}$.
\end{lmm}
\begin{proof}
  The statement follows from standard optimization with Lagrange
  multipliers method and from the straightforward observation that
  given two normalized pure states $\ket{v_0}$ and $\ket{v_1}$, if a
  pure state $\ket{u}$ can be decomposed as follows
  \begin{align*}
    \ket{u} = \braket{v_0}{u}\ket{v_0} + \braket{v_1}{u}\ket{v_1},
  \end{align*}
  then $|\braket{v_0}{u}| = |\braket{v_1}{u}|$.
\end{proof}

Making use of Lemmas \ref{thm:recbound} and \ref{thm:2bound}, we
provide upper and lower bounds on $I_{d+1}^*$ as follows
\begin{align}\label{eq:bound}
  d - 2 + \sqrt{2} \le I_{d+1}^* \le d,
\end{align}
where the second inequality follows from the non discriminability of
$d+1$ states in dimension $d$ (see \cite{GBHA10}).

We now make use of these facts to provide our main result, namely a
lower threshold for the detection efficiency required to reliably
dimension witnessing. We consider the problem of lower bounding the
dimension of a system prepared by a non-characterized source in
Proposition \ref{thm:etac}, as well as the problem of discriminating
between the quantum or classical nature of a source in Proposition
\ref{thm:etad}.

\begin{thm}\label{thm:etac}
  For any $d$ there exists a dimension witnessing setup such that it
  is possible to discriminate between the quantum and classical nature
  of a $d$-dimensional system using POVMs with detection efficiency
  $\eta$ whenever
  \begin{align}\label{eq:etac}
    \eta \ge \eta_{qc} := (d-1) / I_{d+1}.
  \end{align}
  Furthermore one has
  \begin{align}\label{eq:etacbound}
    \frac{d-1}d \le \eta_{\textrm{qc}} \le \frac{d-1}{d - 2 +
      \sqrt{2}}.
  \end{align}
\end{thm}

\begin{proof}
  We provide a constructive proof of the statement. Take $M = d+1$, $K
  = d$, and $N = 3$, and we show that $I_{d+1}$ satisfies the thesis.

  We notice that the maximum value of $I_{d+1}$ attainable with
  classical states is given by $d-1$ \cite{GBHA10}. Then
  $\eta_{\textrm{qc}}$ is the minimum value of the detection
  efficiency such that $I_{d+1}$ can discriminate a quantum system
  from a classical one.

  Due to Lemma \ref{thm:lossy} we have Eq. \eqref{eq:etac}. From
  Eq. \eqref{eq:bound} the lower and upper bounds for
  $\eta_{\textrm{qc}}$ given in Eq. \eqref{eq:etacbound}
  straightforwardly follow.
\end{proof}

Notice that $I_{d+1}$ in Eq. \eqref{eq:etac} can be numerically
evaluated with the techniques discussed in Appendix
\ref{sect:optId}. Figure \ref{fig:etac} plots the value of
$\eta_{\textrm{qc}}$ for different values of the dimension $d$ of the
Hilbert space $\hilb{H}$. The threshold in the detection efficiency
when $d=2$ is $\eta_{\textrm{qc}}=1/\sqrt{2}$, going asymptotically to
$1$ with $d$ as $\sim1+1/d$.
\begin{figure}[htb]
  \includegraphics[width=\columnwidth]{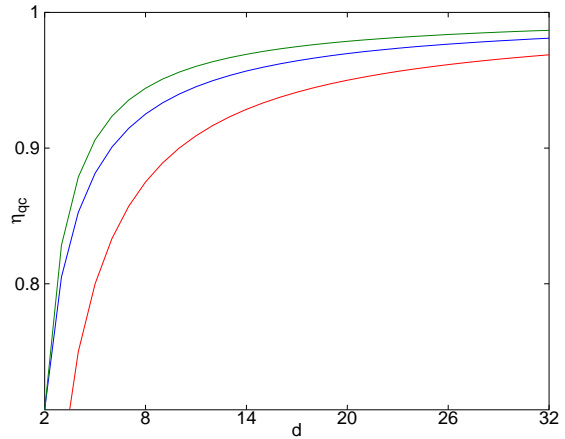}
  \caption{(Color on line) Threshold value (middle line) of the
    detection efficiency $\eta_{\textrm{qc}}$ as in
    Eq. \eqref{eq:etac} as a function of the dimension $d$, obtained
    through numerical optimization of $I_{d+1}$ with Algorithm
    \ref{thm:algo2}. The lower bound (lower line) and upper bound
    (upper line) given by Eq. \eqref{eq:etacbound} are also
    plotted. As expected, the upper bound is tight for $d=2$. The
    detection efficiency $\eta_{\textrm{qc}}$ asymptotically goes to
    $1$ as $d\to\infty$ since its upper and lower bound do the same.}
  \label{fig:etac}
\end{figure}

\begin{thm}\label{thm:etad}
  For any $d$ there exists a dimension witnessing setup such that it
  is possible to lower bound the dimension of a $d+1$-dimensional
  system using POVMs with detection efficiency $\eta$ whenever
  \begin{align}\label{eq:etad}
    \eta \ge \eta_{\textrm{dim}} := I_{d+1} / d.
  \end{align}
  Furthermore one has
  \begin{align}\label{eq:etadbound}
    \eta_{\textrm{dim}} \ge 1 - \frac{2 - \sqrt{2}}d.
  \end{align}
\end{thm}

\begin{proof}
  We provide a constructive proof of the statement. Take $M = d+1$, $K
  = d$, and $N = 3$, and we show that $I_{d+1}$ satisfies the thesis.

  We notice that the maximum value of $I_{d+1}$ attainable in any
  dimension $> d$ is given by $d$ \cite{GBHA10}. Then
  $\eta_{\textrm{dim}}$ is the minimum value of the detection
  efficiency such that $I_{d+1}$ can lower bound the dimension of a
  $d+1$ dimensional system.

  Due to Lemma \ref{thm:lossy} we have Eq. \eqref{eq:etad}. From
  Eq. \eqref{eq:bound} the lower bound to $\eta_{\textrm{dim}}$ given
  by Eq. \eqref{eq:etadbound} straightforwardly follows.
\end{proof}

Notice that $I_{d+1}$ in Eq. \eqref{eq:etad} can be numerically
evaluated with the techniques discussed in Appendix
\ref{sect:optId}. Figure \ref{fig:etad} plots the value of
$\eta_{\textrm{dim}}$ for different values of the dimension $d$ of the
Hilbert space $\hilb{H}$. The threshold in the detection efficiency
when $d=2$ is $\eta_{\textrm{qc}}=1/\sqrt{2}$, going asymptotically to
$1$ with $d$ as $\sim1+1/d$. We notice that $\eta_{\textrm{dim}}$
grows faster than $\eta_{qc}$, thus showing that for fixed dimension,
the discrimination between the quantum or classical nature of the
source is more robust to loss than lower bounding the dimension of the
prepared states.
\begin{figure}[htb]
  \includegraphics[width=\columnwidth]{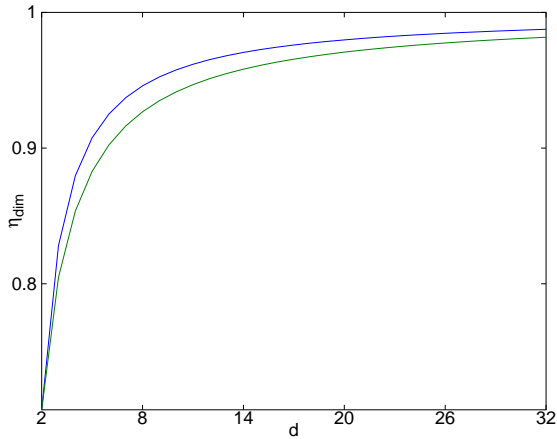}
  \caption{(Color on line) Threshold value (upper line) of the
    detection efficiency $\eta_{\textrm{dim}}$ as in
    Eq. \eqref{eq:etad} as a function of the dimension $d$, obtained
    through numerical optimization of $I_{d+1}$ with Algorithm
    \ref{thm:algo2}. The lower bound (lower line) given by
    Eq. \eqref{eq:etadbound} is also plotted. As expected, the lower
    bound is tight for $d=2$. The detection efficiency
    $\eta_{\textrm{dim}}$ asymptotically goes to $1$ as $d\to\infty$
    since its lower bound does the same (and $\eta_{\textrm{dim}} \le
    1$ is a trivial upper bound).}
  \label{fig:etad}
\end{figure}

\section{Conclusion}\label{sect:conclusion}

In this work we addressed the problem whether a lossy setup can
provide a reliable lower bound on the dimension of a classical or
quantum system. First we provided some relevant properties of the sets
of classical and quantum correlations attainable in a dimension
witnessing setup. Then we introduced analytical and numerical tools to
address the problem of the robustness of DIDWs, and we provided the
amount of loss that can be tolerated in dimension witnessing. The
presented results are of relevance for experimental implementations of
DIDWs, and can be naturally applied to semi-device independent QKD and
QRACs.

We notice that, while we provided analytical proofs of our main
results, i.e. Propositions \ref{thm:etac} and \ref{thm:etad}, their
optimality as a bound relies on numerical evidences. In particular,
they are optimal if the dimension witness $I_{d+1}$ is indeed the most
robust to loss for any $d$, which is suggested by numerical evidence
obtained with the techniques of Appendix \ref{sect:optW} and Appendix
\ref{sect:optId}. Thus, a legitimate question is whether the bounds
provided in Propositions \ref{thm:etac} and \ref{thm:etad} are indeed
optimal. Moreover, it is possible to consider models of loss more
general than the one considered here, e.g. one in which a different
detection efficiency is associated to any POVM.

A natural generalization of the problem of DIDWs, in the ideal as well
as in the lossy scenario, is that in the absence of correlations
between the preparations and the measurements. In this case, as
discussed in this work, the relevant sets of correlations are $\set Q$
and $\set C$, which are non-convex as shown in Section
\ref{sect:didw}. The non convexity of the relevant sets allows the
exploitation of non-linear witnesses - as opposed to what we did in
the present work. An intriguing but still open question is whether
there are situations in which this exploitation allows to dimension
witness for any non-null value of the detection efficiency.

Another natural generalization of the problem of DIDWs is that of
entangled assisted DIDWs, namely when entanglement is allowed to be
shared between the preparing device on Alice's side and the measuring
device on Bob's side. This problem is similar to that of super-dense
coding~\cite{BW92}. Consider again Fig. \ref{fig:setup}. In the
simplest super-dense coding scenario, Alice presses one button out of
$M=4$, while Bob always performs the same POVM ($K=1$) obtaining one
out of $N=4$ outcomes. The dimension of the Hilbert space $\hilb{H}$
is $\dim(\hilb{H}) = 2$, but a pair of maximally entangled qubits is
shared between the parties. In this case, the results of~\cite{BW92}
imply that a classical system of dimension $4$ (quart) can be sent
from Alice to Bob by sending a qubit (corresponding to half of the
entangled pair).

Consider the general scenario where now the two parties are allowed to
share entangled particles. The super-dense coding protocol
automatically ensures that by sending a qubit Alice and Bob can always
achieve the same value of any DIDW as attained by a classical
quart. Remarkably, the super-dense coding protocol turns out not to be
optimal, as we identified more complex protocols beating it. In
particular, we found a $(M=4,K=2,N=4)$ situation for which, upon
performing unitary operations on her part of the entangled pair and
subsequently sending it to Bob, Alice can achieve correlations that
can not be reproduced upon sending a quart. This thus proves the
existence of communication contexts in which sending half of a
maximally entangled pair is a more powerful resource than a classical
quart. This observation is analogous to that done in \cite{PZ10},
where it was shown that entangled assisted QRACs (where an entangled
pair of qubits is shared between the parties) outperform the best of
known QRACs.  For these reasons we believe that the problem of
entangled assisted DIDWs deserves further investigation.

\section*{Acknowledgments}

We are grateful to Nicolas Brunner, Stefano Facchini, and Marcin
Paw{\l}owski for very useful discussions and suggestions. M. D. thanks
Anne Gstottner and the Human Resources staff at ICFO for their
invaluable support. This work was funded by the Spanish FIS2010-14830
Project and Generalitat de Catalunya, the European PERCENT ERC
Starting Grant and Q-Essence Project, and the Japanese Society for the
Promotion of Science (JSPS).

\appendix

\section{Numerical optimization of dimension witnesses}
\label{sect:optW}

Given a linear dimension witness $W$ the following algorithm converges
to a local maximum of $W(R,P)$.
\begin{algo}\label{thm:algo1}
  For any set $R^{(0)} = \{\psi_i^{(0)}\}$ of pure states and any set
  $P^{(0)} = \{\Pi_k^{(0)}\}$ of POVMs $\Pi_k^{(0)} =
  \{\Pi_k^{j,(0)}\}$,
  \begin{enumerate}
    \item let $\ket{\bar{\psi}_i^{(n+1)}} = \left[ (1-\epsilon)
      \identity + \epsilon \sum_{j,k} c_{i,j,k} \Pi_k^{j,(n)} \right]
      \ket{\psi_i^{(n)}}$,
    \item let $\bar{\Pi}_k^{j,(n+1)} = \left\{\left[ (1-\epsilon)
      \identity + \epsilon \sum_i c_{i,j,k} \psi_i^{(n)} \right]
      \sqrt{\Pi_k^{j,(n)}}\right\}^2$,
    \item normalize $\ket{\psi_i^{(n+1)}} =
      ||\bar{\psi}_i^{(n+1)}||^{-1/2} \ket{\bar{\psi}_i^{(n+1)}}$,
    \item normalize $\Pi_k^{j,(n+1)} = S_k^{-\frac12}
      \bar{\Pi}_k^{j,(n+1)} S_k^{-\frac12}$ with $S_k = \sum_j
      \bar{\Pi}_k^{j,(n+1)}$.
  \end{enumerate}
\end{algo}

As for all steepest-ascent algorithm, there is no protection against
the possibility of convergence toward a local, rather than a global,
maximum. Hence one should run the algorithm for different initial
ensembles in order to get some confidence that the observed maximum is
the global maximum (although this can never be guaranteed with
certainty). Any initial set of states and any initial set of POVMs can
be used as a starting point, except for a subset corresponding to
minima of $W(R,P)$. These minima are unstable fix-points of the
iteration, so even small perturbations let the iteration converge to
some maxima. The parameter $\epsilon$ controls the length of each
iterative step, so for $\epsilon$ too large, an overshooting can
occur. This can be kept under control by evaluating $W(R,P)$ at the
end of each step: if it decreases instead of increasing, we are warned
that we have taken $\epsilon$ too large.

Referring to Fig. \ref{fig:setup}, the simplest non-trivial scenario
one can consider is the one with $M=3$ preparations and $K=2$ POVMs
each with $N=3$ outcomes, one of which corresponding to no-click
event. In this case one has several tight classical DIDWs. Applying
Algorithm \ref{thm:algo1} we verified that among them the most robust
to loss is given by Eq. \eqref{eq:ldidw} with coefficients given by
Eq. \eqref{eq:Id}.

\section{Numerical optimization of $I_{d+1}$}
\label{sect:optId}

The following Lemma proves that the POVMs maximizing $I_{d+1}$ for any
dimension $d$ are such that one of their elements is a projector on a
pure state, thus generalizing a result from \cite{Mas05}.
\begin{lmm}\label{thm:purepovms}
  For any dimension $d$, the maximum of $I_{d+1}$ is achieved by a set
  $P = \{\Pi_k\}$ of POVMs $\Pi_k = \{\Pi_k^j\}$ with $\Pi_k^1$ a
  projector with $\rank \Pi_k^1 = 1$ for any $k$.
\end{lmm}
\begin{proof}
  For any fixed set $R = \{ \psi_i \}$ of pure states define $A_k := -
  \sum_{i\neq k} \psi_i$, $B := \psi_k$, and $X_k := A_k + B_k$. Then
  clearly $A_k \le 0$, $B_k \ge 0$ and $\rank B_k = 1$ for any
  $k$. From Eq. \eqref{eq:ldidw} it follows immediately that the
  optimal set $P^* = \{ \Pi_k^* \}$ of POVMs $\Pi_k^* = \{ \Pi^{*j}_k
  \}$ is such that $\Pi_k^{*1} = \arg\min_{\Pi_k^1} \Tr[X
    \Pi_k^1]$. The optimum of $I_{d+1}$ is achieved when $\Pi_k^1$ is
  the sum of the eigenvectors of $X_k$ corresponding to positive
  eigenvalues.

  Upon denoting with $\lambda_1(A_k) \ge \dots \ge \lambda_n(A_k)$ the
  eigenvalues of $A_k$, the Weyl inequality (see for example
  \cite{Bha06}) $\lambda_1(X_k) \le \lambda_1(A_k) + \lambda_n(B_k)$
  holds for any $n$. Since $\lambda_1(A_k) \le 0$ and $\lambda_n(B_k)
  = 0$ for any $k$ and for any $n \neq 0$, the thesis follows
  immediately.
\end{proof}

Algorithm \ref{thm:algo1} can be simplified using Lemma
\ref{thm:purepovms}. The following algorithm converges to a local
maximum of $I_{d+1}$.
\begin{algo}\label{thm:algo2}
  For any set $R^{(0)} = \{\psi_i^{(0)}\}$ of pure states and any set
  $P^{(0)} = \{\Pi_k^{(0)}\}$ of POVMs $\Pi_k^{(0)} =
  \{\Pi_k^{j,(0)}\}$,
  \begin{enumerate}
    \item let $\ket{\bar{\psi}_i^{(n+1)}} = \ket{\psi_i^{(n)}} +
      \epsilon \sum_{j,k} c_{i,j,k} \braket{\pi_k^{(n)}}{\psi_i^{(n)}}
      \ket{\pi_k^{(n)}}$,
    \item let $\ket{\bar{\pi}_k^{(n+1)}} = \ket{\pi_k^{(n)}} +
      \epsilon \sum_{i,k} c_{i,j,k} \braket{\psi_i^{(n)}}{\pi_k^{(n)}}
      \ket{\psi_i^{(n)}}$,
    \item normalize $\ket{\psi_i^{(n+1)}} =
      ||\bar{\psi}_i^{(n+1)}||^{-1/2} \ket{\bar{\psi}_i^{(n+1)}}$,
    \item normalize $\ket{\pi_k^{(n+1)}} =
      ||\bar{\pi}_k^{(n+1)}||^{-1/2} \ket{\bar{\pi}_k^{(n+1)}}$.
  \end{enumerate}
\end{algo}

The same remarks made about Algorithm \ref{thm:algo1} hold true for
Algorithm \ref{thm:algo2}. Nevertheless, we verified that in practical
applications Algorithm \ref{thm:algo2} always seems to converge to a
global, not a local maximum. This can be explained considering that
without loss of generality it optimizes over a smaller set of POVMs
when compared to Algorithm \ref{thm:algo1}. Moreover, we noticed that
the optimal sets of states and POVMs are real, namely there exists a
basis with respect to which states and POVM elements have all real
matrix entries. A similar observation was done in \cite{FFW11} in the
context of Bell's inequalities.


\begin{thebibliography}{}
\bibitem{note1} This is not the case in the hybrid scenario where
  different types of detectors (e.g. photodetectors and homodyne
  measurements) are used. A similar scenario was proposed for example
  in the context of Bell inequalities \cite{CBSSS11}.
\end{thebibliography}
\end{document}